\documentclass[conference]{IEEEtran}
\IEEEoverridecommandlockouts
\usepackage{cite}
\usepackage{amsmath,amssymb,amsfonts}
\usepackage{algorithmic}
\usepackage{graphicx}
\usepackage{textcomp}
\usepackage{xcolor}

\usepackage{bm}
\usepackage{amsthm}
\usepackage[linesnumbered,ruled,vlined]{algorithm2e}
\usepackage{multirow}
\usepackage{caption}
\usepackage{hyperref}

\def\BibTeX{{\rm B\kern-.05em{\sc i\kern-.025em b}\kern-.08em
    T\kern-.1667em\lower.7ex\hbox{E}\kern-.125emX}}

\newtheorem{remark}{Remark}
\newcommand{\tensor}{\otimes}

\newcommand{\ket}[1]{|#1\rangle}
\newcommand{\bra}[1]{\langle #1 |}
\newcommand{\tr}{\textnormal{tr}}

\newcommand{\parity}{\textnormal{Par}}
\newcommand{\std}{\textnormal{std}}
\newcommand{\diag}{\textnormal{diag}}
\newcommand{\Cov}{\textnormal{Cov}}
\newtheorem{definition}{Definition}
\newtheorem{proposition}{Proposition}
\newcommand{\Var}{\textnormal{Var}}
\newcommand{\inv}{^{-1}}

\begin{document}

\title{Online Detection of Golden Circuit Cutting Points\\
\thanks{This research was supported in part by NSF Awards 2216923, 2117439, and 2230111.}
}

\author{\IEEEauthorblockN{Daniel T. Chen}
\IEEEauthorblockA{\textit{Case Western Reserve University}\\
Cleveland, OH, USA \\
txc461@case.edu}
\and
\IEEEauthorblockN{Ethan H. Hansen}
\IEEEauthorblockA{\textit{Case Western Reserve University}\\
Cleveland, OH, USA \\
ehh50@case.edu}
\and
\IEEEauthorblockN{Xinpeng Li}
\IEEEauthorblockA{\textit{Case Western Reserve University}\\
Cleveland, OH, USA \\
xxl1337@case.edu}
\and
\IEEEauthorblockN{Aaron Orenstein}
\IEEEauthorblockA{\textit{Case Western Reserve University}\\
Cleveland, OH, USA \\
aao62@case.edu}
\and
\IEEEauthorblockN{Vinooth Kulkarni}
\IEEEauthorblockA{\textit{Case Western Reserve University}\\
Cleveland, OH, USA \\
vxk285@case.edu}
\and
\IEEEauthorblockN{Vipin Chaudhary}
\IEEEauthorblockA{\textit{Case Western Reserve University}\\
Cleveland, OH, USA \\
vxc204@case.edu}
\and
\IEEEauthorblockN{Qiang Guan}
\IEEEauthorblockA{\textit{Kent State University}\\
Kent, OH, USA \\
qguan@kent.edu}
\and
\IEEEauthorblockN{Ji Liu}
\IEEEauthorblockA{\textit{Argonne National Laboratory}\\
Lemont, IL, USA \\
ji.liu@anl.gov}
\and
\IEEEauthorblockN{Yang Zhang}
\IEEEauthorblockA{\textit{University of Illinois Urbana-Champaign}\\
Champaign, IL, USA \\
yzhangnd@illinois.edu}
\and
\IEEEauthorblockN{Shuai Xu}
\IEEEauthorblockA{\textit{Case Western Reserve University}\\
Cleveland, OH, USA \\
sxx214@case.edu}
}

\maketitle

\begin{abstract}
    Quantum circuit cutting has emerged as a promising method for simulating large quantum circuits using a collection of small quantum machines.
    Running low-qubit circuit ``fragments'' not only overcomes the size limitation of near-term hardware, but it also increases the fidelity of the simulation.
    However, reconstructing measurement statistics requires computational resources---both classical and quantum---that grow exponentially with the number of cuts. 
    In this manuscript, we introduce the concept of a golden cutting point, which identifies unnecessary basis components during reconstruction and avoids related downstream computation. 
    We propose a hypothesis-testing scheme for identifying golden cutting points, and provide robustness results in the case of the test failing with low probability.
    Lastly, we demonstrate the applicability of our method on Qiskit's Aer simulator and observe a reduced wall time from identifying and avoiding obsolete measurements.
\end{abstract}


\section{Introduction}
\label{sec:intro}

Quantum circuit cutting refers to the method of splitting quantum circuits into a set of small independent circuit fragments~\cite{Peng_2020}. Using circuit cutting methods, large quantum circuits can be simulated by a collection of smaller machines, barring some addition classical computing resources. Moreover, it was also empirically shown that cutting the circuit reduces the affect of noise~\cite{ayral2020quantum, ayral2021quantum} and can be used for error mitigation~\cite{liu2022classical}. There has also been work on properly accounting for statistical shot noise~\cite{Perlin_2020} and adaptation to specific problems such as combinatorial optimization~\cite{saleem2021quantum}. Thus, this technique holds great promise for resolving many practical issues with utilizing quantum hardware, particularly in the NISQ era~\cite{preskill2018quantum}.

However, circuit cutting suffers greatly in practice as the runtime grows exponentially with the number of cuts. Akin to quantum state tomography, circuit cutting works by classically tracking all quantum degrees of freedom at the cut locations. Thus, exponential scaling comes naturally as the quantum state of interest grows. There have been many efforts to reduce this cost through randomized measurements~\cite{Lowe_2022, chen2022quantum}, classical sampling~\cite{Chen_2022}, and variational optimization~\cite{Uchehara2022}. Nonetheless, the exponential growth in runtime is unlikely to vanish without imposing structural assumptions on the circuit. Alternatively, finding applications of circuit cutting that avoid the scaling issue, as demonstrated in~\cite{liu2022classical}, also remains a problem of interest.

In this manuscript, we build upon our previous work~\cite{chen2023efficient} and propose an algorithm for online detection of neglectable basis elements during circuit cutting. In~\cite{chen2023efficient}, we showed that some reconstruction procedures can be sped up if we impose extra assumptions on the circuit---namely, whether a basis element can be neglected. However, such assumptions cannot be easily detected \textit{a priori}. Thus, we propose a hypothesis testing scheme at each cut location that, at no additional cost in run time, identifies whether there is statistically significant evidence against the assumption being true. We empirically demonstrate the viability of our method on the Qiskit simulator~\cite{Qiskit}, and examine scaling effects with respect to important algorithm parameters. 

The outline of the paper is as follows. Section \ref{sec:theory} re-derives the circuit cutting in the general bipartition case and introduces the concept of a ``golden cut''---a cut location that has basis elements which can be neglected. We also show the algorithm for detecting golden cuts as well as their statistical properties. In Section \ref{sec:experiments}, we demonstrate the applicability of our method in a simple, one-cut case. Meanwhile, we explore additional properties of the proposed algorithm under varying parameters. Lastly, we discuss some future directions in Section \ref{sec:conclusion}.

\section{Theory}
\label{sec:theory}

\begin{figure*}[t]
    \centering
    \includegraphics[width=0.9\linewidth]{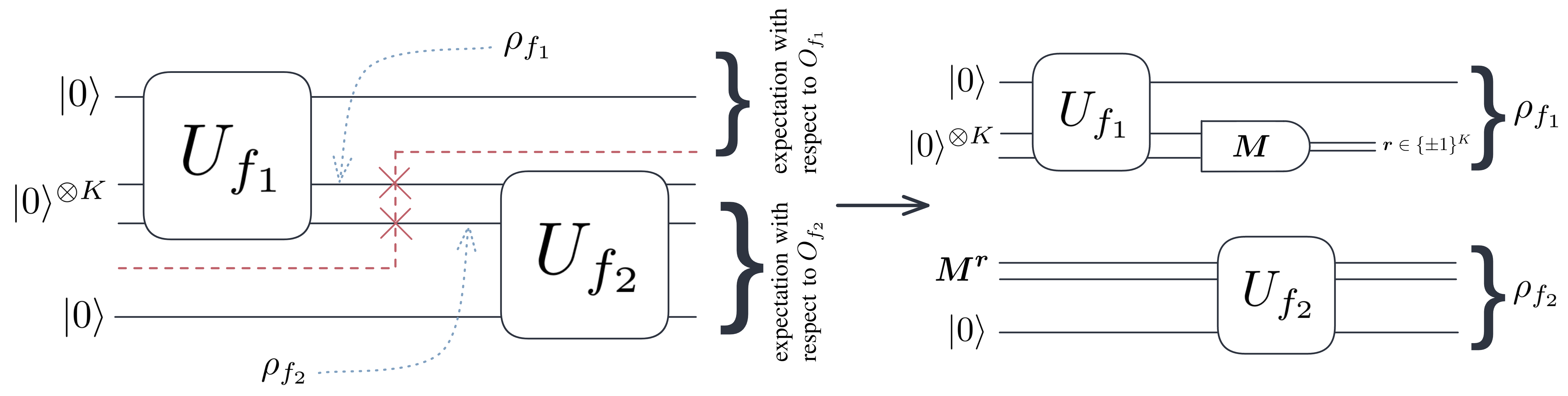}
    \caption{An example of circuit bipartition using $K$ cuts. The circuit is comprised of arbitrary unitary gates $U_{f_1}$ and $U_{f_2}$ such that the middle $K$ qubits are possibly entangled. The circuit can be split into two fragments by performing $K$ cuts between the two gates, and classically recombining measurement outcomes in basis $\bm M$ and preparing eigenstates of $\bm M$ for a collection of operators that forms a basis over density matrices.}
    \label{fig:cut_diagram}
\end{figure*}

We start the section by briefly introducing circuit cutting only for the bipartite case. Readers are encouraged to consult~\cite{Peng_2020, Perlin_2020, chen2022quantum} for more detailed derivations and/or general treatment of circuit cutting. Section \ref{sec:bipartite} ends with the definition of golden cuts---circuit structure that induces neglectable basis elements. Then, Section \ref{sec:online-detect} derives an online detection algorithm for these structures.

\subsection{Circuit Bipartition} \label{sec:bipartite}

Suppose an $N$-qubit quantum circuit induces a state $\rho$. Further suppose we want to perform $K$ cuts on a quantum circuit to divide it into two fragments, $f_1$ and $f_2$ (see Figure \ref{fig:cut_diagram}). Indexing each cut by an integer in $[K]$, then we can write the cutting scheme as an injective function $c: [K] \to [N]$ that maps a cut to the respective qubit being cut. The premise of circuit cutting is to rewrite $\rho$ in terms of the states induced by the circuit fragments, $\rho_{f_1}$ and $\rho_{f_2}$, albeit with some parameterization.

As first shown in Peng \textit{et al.}~\cite{Peng_2020}, this can be done via expanding $\rho$ at the location of the cut. That is, for a basis set over $2 \times 2$ Hermitian matrices, $\mathcal B =\{I, X, Y, Z\}$, the following decomposition holds:
\begin{align}
    \rho = \frac{1}{2^K} \sum_{\bm M \in \mathcal B^K} \rho_{f_1}(\bm M) \tensor \rho_{f_2}(\bm M)
\end{align}
where 
\begin{align}
    \bm M = \begin{pmatrix} M_{c(1)} & M_{c(2)} & \dots & M_{c(K)} \end{pmatrix}, ~~ M_{c(i)} \in \mathcal B.
\end{align} 

The state of each fragment $\rho_{f_i}$, $i = 1,2, \dots$ is parameterized by an operator $\bm M$ and depends on the particular gates contained in the circuit. Let $U_{f_i}$ denote the unitary operation induced by the quantum gates on each fragment, and let $\ket{\bm 0}$ be an $N_{f_1}$-qubit ``zero'' state. Then, we can write $\rho_{f_i}$ as the following (up to appropriate qubit permutation):
\begin{align}
    \rho_{f_1}(\bm M) &= \tr_{c(1), \dots, c(K)} \left( \bigotimes_{i \in [K]} M_{c(i)} U_{f_1} \ket{\bm 0}\bra{\bm 0} U_{f_1}^\dagger \right),\\
    \rho_{f_2}(\bm M) &= U_{f_2} \left( \bigotimes_{i \in [K]} M_{c(i)} \tensor \ket{\bm 0}\bra{\bm 0} \right) U_{f_2}^\dagger.
\end{align}
The choice of basis is arbitrary, and we chose the normalized Pauli basis for simplicity. Note that the above equation lacks a physical interpretation as elements in $\mathcal B$ are traceless (except for $I$) and hence, are not quantum states.


To resolve this issue, we note that each operator $M$ admits spectral decomposition. Letting 
\begin{align}
    \bm r = \begin{pmatrix} r_{c(1)} & r_{c(2)} & \dots & r_{c(K)} \end{pmatrix} \in \{\pm 1\}^K
\end{align}
be a tuple of eigenvalues, we define 
\begin{align}
    \bm M^{\bm r} = \begin{pmatrix} M_{c(1)}^{r(1)} & M_{c(2)}^{r(2)} & \dots & M_{c(K)}^{r(K)} \end{pmatrix}
\end{align}
to be the $\bm r$-th eigenstate of operator $\bm M$. Let $\bm s \in \{\pm 1\}^K$ and $\bm M^{\bm s}$ be similarly defined. Applying this decomposition gives the reconstruction formula in the bipartition case: 
\begin{align}
    \rho = \frac{1}{2^K} \sum_{\substack{\bm M \in \mathcal B^K, \\ \bm r, \bm s \in \{\pm 1\}^K}} \parity(\bm r)~\parity(\bm s) ~ \rho_{f_1}(\bm M^{\bm r}) \tensor \rho_{f_2}(\bm M ^{\bm s})
\end{align}
where $\parity(\bm r)$ denotes the parity of a string of eigenvalues, i.e., $\parity(\bm r) = \prod_i r_i$. The formula above lends itself to a measure-and-prepare scheme for realizing quantum circuit cutting: for each basis element $\bm M$, we measure the upstream circuit in the basis, prepare the downstream circuit into the eigenstates of the same basis, and reweight the outcome of the downstream circuit by the probability of observing the respective eigenstate upon measuring the upstream fragment.

Alternatively, for any desired quantum observable $O$, suppose the operator can be decomposed to accommodate the two fragments, i.e., $O = O_{f_1} \tensor O_{f_2}$ up to appropriate permutation of qubit indices.  Then, we can arrive at an expression for the expectation of the uncut circuit in terms of the fragments $\rho_{f_i}$ and their respective observables $O_{f_i}$: 
\begin{align}
    & \tr(O\rho) \nonumber \\
    &= \frac{1}{2^K} \sum_{\bm M, \bm r, \bm s} \parity(\bm r) \parity(\bm s) \tr \left( \left(O_{f_1} \tensor O_{f_2} \right) \left( \rho_{f_1} \tensor \rho_{f_2} \right) \right) \\
    & =\frac{1}{2^K} \sum_{\bm M, \bm r, \bm s} \parity(\bm r) \parity(\bm s)~ \tr \left( O_{f_1} \rho_{f_1}\right) \tr \left( O_{f_2} \rho_{f_2} \right) \label{eq:bipartite-cut-expval}
\end{align}
where we implicitly apply $\bm M$ to $\rho_{f_1}$ and $\rho_{f_2}$. 

Note that the decomposition assumption $O = O_{f_1} \tensor O_{f_2}$ is without loss of generality. For any choice of Hermitian operator $O$, one can expand it with respect to Pauli strings, i.e.,
\begin{align}
    O = \sum_{\bm S \in \mathcal B^{N}} a_S~ \bm S \label{eq:pauli-decomp}
\end{align}
for some set of real coefficients $\{a_S\}$. So, by linearity of the trace operator, we obtain a generalized expression for the expectation:
\begin{align}
    &\tr(O\rho) \nonumber \\
    &= \frac{1}{2^K} \sum_{\bm S, \bm M, \bm r, \bm s} 
    a_S \parity(\bm r) \parity(\bm s) \tr \left( \bm S_{f_1} \rho_{f_1}\right) \tr \left( \bm S_{f_2} \rho_{f_2} \right)
\end{align}
where $\bm S_{f_1}$ and $\bm S_{f_2}$ are the Pauli strings separated according to the circuit cutting scheme, i.e., $\bm S = \bm S_{f_1} \tensor \bm S_{f_2}$ under appropriate qubit permutations. 

We now formally define the \textit{golden circuit cutting point}.
\begin{definition}
Consider an $N$-qubit circuit amenable to bipartition with $K$ cuts. We're interested in the expectation of the circuit-induced state with respect to some quantum observable $O = O_{f_1} \tensor O_{f_2}$. The cutting scheme admits a \emph{golden cutting point} if there exists $\bm M_* \in \mathcal B^K$ such that 
\begin{align}
    \sum_{\bm r \in \{\pm 1\}^K} \parity(\bm r) ~\tr \left( O_{f_1} \rho_{f_1}(\bm M_*^{\bm r}) \right) = 0 \label{def:golden-cut}
\end{align}
\end{definition}

More simply put, a golden cutting point refers to the existence of a basis element that leads to systematic cancellations. Golden cutting points neither necessarily exists nor are unique. For each such basis element, one does not need to execute the circuit downstream of the cut with initialization corresponding to the neglected basis. 

Golden cutting points can be constructed via circuit design by restricting the set of rotations allowed prior to cutting so long as the structure of the quantum circuit permits. However, one should not expect such a property to hold for an arbitrary algorithm. Thus, we propose an ``online'' scheme for detecting the existence of golden cutting points and establish robustness of misidentifying golden cuts. 

\subsection{Identifying Golden Cutting Points}\label{sec:online-detect}

With no knowledge of the existence of golden cutting points, one must execute each of the $4^K$ upstream circuits and another $4^K$ downstream circuits ($f_1$ and $f_2$ respectively in the bipartition case). To detect golden cutting points in the absence of \textit{a priori} knowledge, we propose to conduct a hypothesis test for each of the $4^K$ upstream circuits, determine whether there is statistically significant evidence for the existence of a golden cutting point, then run the corresponding downstream circuit. 

Denote 
\begin{align}
    \tau = \sum_{\bm r \in \{\pm 1\}^K} \parity(\bm r) ~\tr \left( O_{f_1} \rho_{f_1}(\bm M_*^{\bm r}) \right)
\end{align}
as the quantity we want to verify magnitude of. Inheriting the bipartition assumption from the previous section, we can rewrite $\tau$ as estimating the expectation of a larger observable
\begin{align}
    \tau = \tr \left( \left(O_{f_1} \tensor \bm M_* \right) U_{f_1} \ket{\bm 0}\bra{\bm 0} U_{f_1}^\dagger \right).
\end{align}
where, again, $U_{f_1}\ket{\bm 0}$ is the state induced by the upstream fragment. Writing $\tau$ in this form allows us to employ standard techniques for estimating quantum observables. 

Assume for convenience that $O_{f_1}$ is a Pauli-string. To estimate the expectation, we measure each qubit in the respective Pauli basis (by performing a rotation $V$) $m$ times and obtain an ensemble of bitstring samples $\{\hat b_i\}_{i=1}^m$. Therefore, we can estimate $\tau$ be constructing
\begin{align}
    \hat \tau = \frac{1}{m} \sum_{i=1}^m \bra{\hat b_i} V^\dagger (O_{f_1} \tensor \bm M) V \ket{\hat b_i} \label{eq:estimator}.
\end{align}
Alternatively, one can think of estimating the distribution of strings of eigenvalues (which we'll call eigenstrings for short) induced by the measurements. Write $p_{b}$ for the probability of obtaining eigenstring $\bm b$, and $\hat p_b$ for the empirical probability. Moreover, let $\bm p$ and $\hat{\bm p}$ denote the vector of probabilities. Then, we can sum the parity of each eigenstring weighted by the (empirical) probability to arrive at an alternative expression for the estimator: 
\begin{align}
    \hat \tau = \sum_{\bm b \in \{\pm 1\}^{N_{f_1}}} \parity(\bm b)~ \hat p_b. \label{eq:pauli-exp}
\end{align}

The proposition below establishes the standard error and asymptotic normality, which are convenient for hypothesis testing.

\begin{proposition}[Asymptotic Normality]
    Given a circuit amenable to the bipartite circuit cutting scheme (cf.~\ref{sec:bipartite}), let $\hat \tau$ be the estimator of $\tau$ expressed in Equation \ref{eq:estimator} and let $O_{f_1}$ admit decomposition as in Equation \ref{eq:pauli-decomp}. Then, $\hat \tau$ is asymptotically normal, i.e.,
    \begin{align}
        \frac{\hat \tau - \tau}{\std(\hat \tau)} \to \mathcal N(0,1)
    \end{align}
    where the standard deviation of the estimator is expressed as
    \begin{align}
        \std(\hat \tau) = \left( \sum_{S \in \mathcal B^{N_{f_1}}} \frac{a_S^2}{N}~ \bm \chi^\intercal (\diag(\hat{\bm p_S}) - \hat{\bm p_S}\hat{\bm p_S}^\intercal) \bm \chi \right)^{1/2}
        \label{eq:std-error}
    \end{align}
    and $\bm \chi$ is the vector of parities, i.e., $\bm \chi_b = \parity(b).$
\end{proposition}

\begin{proof}
We first consider the case were $O_{f_1}$ is a Pauli string $\bm S$, then proceed to generalize to arbitrary quantum observables. 

Estimating the expectation of the Pauli string $\bm S$ with $m$ shots, using the formalism presented in Equation \ref{eq:pauli-exp}, gives the standard error
\begin{align}
    \Var(\hat \tau) &= \Cov \left( \sum_b \hat p_b~ \parity(b), \sum_{b'} \hat p_{b'}~ \parity(b') \right) \\
    &= \sum_{b,b'} \Cov(\hat p_b, \hat p_{b'}) \parity(b) \parity(b') \\
    &= \frac{1}{m}  \left( \sum_{b=b'} p_b(1-p_b)  - \sum_{b \neq b'} p_b p_{b'} \parity(b) \parity(b') \right) 
\end{align}
where the third equality follows from the covariance of multinomial distributions. Using the empirical quantities for each $p_b$ and writing in matrix form gives the estimated standard error
\begin{align}
    \std(\hat \tau) = \sqrt{ \frac{1}{m} \bm \chi^\intercal (\diag(\hat{\bm p_S}) - \hat{\bm p_S}\hat{\bm p_S}^\intercal) \bm \chi}.
\end{align}

Consider the decomposition in Equation \ref{eq:pauli-decomp}. As the estimation procedure runs independently, variances add. Hence, we arrive at the form in Equation \ref{eq:std-error}.

Lastly, asymptotic normality is established by the equivalence formulation presented in Equation \ref{eq:estimator} and Equation \ref{eq:pauli-exp}. In the form of Equation \ref{eq:estimator}, we can express $\hat \tau$ as the average over independent samples. In combination with finiteness of $O_{f_1}$, the Central Limit Theorem holds, implying asymptotic normality of $\hat \tau$.
\end{proof}

Using the above proposition, we can deduce an algorithm for detecting golden cutting points. For each basis element in $\mathcal B^K$, we will compute $\hat \tau$ and perform a statistical test for whether $\tau \neq 0$. If we've gathered statistically significant evidence for $\tau$ being non-zero, then we would run the downstream fragment parameterized by the respective basis element. On the other hand, if $\hat \tau$ is sufficiently close to zero, then we classify the cut as a golden cutting point and proceed without running the corresponding downstream fragment. Using $\Phi$ to denote the CDF of a standard Gaussian, we summarize the above procedure in Algorithm \ref{alg:online-detection}.

\begin{algorithm}
\caption{Online detection of golden cutting points}\label{alg:online-detection}
\KwIn{fragments $f_1$ and $f_2$, observable $O_{f_1}$, $O_{f_2}$, significance level $\alpha \in (0,1)$}
\KwOut{Expectation $\tr(O\rho)$}
\For{$\bm M \in \mathcal B^K$}{
Compute $\hat \tau$ using Eqn.~\ref{eq:pauli-exp} \\
\If{$|\hat \tau| > \Phi\inv (1-\alpha) \cdot \std(\hat \tau)$} {
Reject the hypothesis and compute $\parity(\bm s)~\tr(O_{f_2} \rho_{f_2}(\bm M^{\bm s})$ for all $s \in \{-1,+1\}^K$
}
\Else{
Fail to reject the hypothesis and set quantities related to $\bm M$ to zero
}}
Reconstruct the full expectation from fragment data using Eqn. \ref{eq:bipartite-cut-expval}
\end{algorithm}

\begin{figure*}[t]
    \centering
    \includegraphics[width=0.8\linewidth]{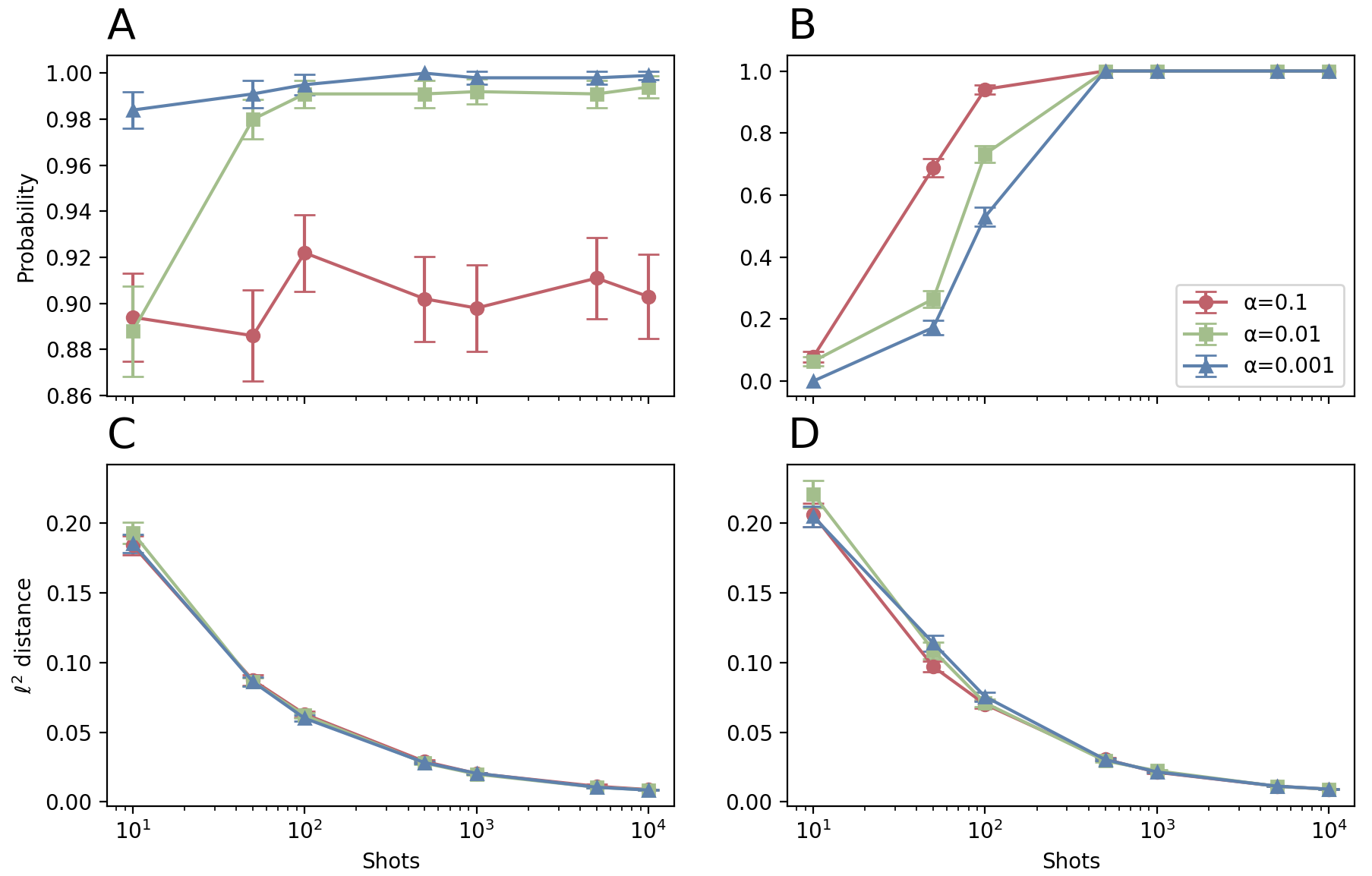}
    \caption{Behavior of Algorithm \ref{alg:online-detection} for varied shots and $\alpha$ levels averaged over 1000 independent trials. The rate of true positives (\textbf{A}) is consistent with the specified $\alpha$, and the true negative rate (\textbf{B}) converges to one as the number of shots increases. The reconstruction error monotonically vanishes both in the presence (\textbf{C}) and the absence (\textbf{D}) of golden cuts.}
    \label{fig:multi_plot}
\end{figure*}


While we can control the rate of correctly identifying golden cutting points by the significance level $\alpha$, we would also like to derive a way of controlling the rate of false negatives. In fact, falsely identifying a non-golden cutting point as golden is more problematic than falsely identifying golden as non-golden. This is because in the latter case the reconstructed bitstring distribution will likely still be within acceptance error ranges, but that's not the case for the former. Thus, we hope to lower the probability of false negatives by taking a sufficient number of shots during the estimation procedure.

\begin{remark}\label{thm:remark}
As we have showed that the estimator converges weakly to a normal distribution, we will assume Gaussianity to facilitate analysis. Suppose $\hat \tau \sim \mathcal N(\tau, b^2/m)$ where $b = \sqrt{m}\cdot \std(\hat \tau)$. We hope that for a basis element where $\tau > \epsilon > 0$, the hypothesis testing scheme would reject it, perhaps with a small probability of error $\delta$. By the Chernoff bound, we know that 
\begin{align}
    \Pr(|\hat \tau - \tau| > \epsilon) \leq 2e^{-m\epsilon^2/2b^2} = \delta.
\end{align}
Thus, to estimate $\tau$ to any desired $\epsilon$ accuracy with probability $1 - \delta$, we need
\begin{align}
    m \geq \frac{2b^2}{\epsilon^2} \log \frac{2}{\delta}
\end{align}
measurements. Since $b$ is not known a priori, we can upper bound it by  
\begin{align}
    \bm 1^\intercal (\diag(\bm q) + \bm q \bm q^\intercal) \bm 1.
\end{align}
The distribution $\bm q$ that maximizes the above quantity is the uniform distribution. Thus, we can arrive at a definitive upper bound $b \leq \frac{3}{2}(1-2^{-N_{f_1}})$. 

Alternatively, one can interpret the proposed sample complexity as accepting an $\epsilon$ margin for Equation \ref{def:golden-cut} in the sense that we declare a cut is golden if $|\hat \tau| < \epsilon$, thereby accepting an additional additive error of magnitude $\mathcal O(\epsilon)$ to the result of the circuit reconstruction. On the other hand, if $\tau > \epsilon$, we wish to identify it with probability $1 - \delta$. Note that in the limit of $m \to \infty$, variance vanishes and the true positive and negative rates approach one.
\end{remark}

\section{Experiments}
\label{sec:experiments}

In this section, we numerically demonstrate the applicability of Algorithm \ref{alg:online-detection}. The algorithm was implemented in Qiskit and executed on the Aer simulator. We examine its statistical (Section \ref{sec:numerics-stats}) and runtime (Section \ref{sec:numerics-time}) properties through studying its dependency on the number of shots and the significance level $\alpha$.

\subsection{Statistical Analysis}
\label{sec:numerics-stats}

As is standard in analyzing binary decisions, we analyzed our statistical test by providing instances when the null hypothesis is true and when it's false. Specifically, we provided circuits either with or without golden cuts, and observed the probability of correctly identifying the existence (or lack) of golden cuts.

We consider a simple three-qubit circuit amenable to cutting on the second qubit. First, we generated a circuit containing a golden cut by appending two $R_X$ gates on the first and second qubits, then an $R_Y$ gate only on the first qubit---this is $U_{f_1}$ from Figure \ref{fig:cut_diagram}, and we let $K=1$. Rotation angle $\theta$ was set to a value far from zero ($\theta = 0.5$) to ensure $X$ would be the only golden cutting axis and to focus only on statistical shot noise. To generate a circuit known to not contain a golden cutting point, we applied the same procedure then appended an additional $R_Y$ gate on the second qubit (the qubit being cut) as well. Finally, we generated $U_{f_2}$ randomly across qubits 2 and 3 using Qiskit's \verb|random_circuit| function. Once the circuit was constructed, for each shot count-$\alpha$ pair, we repeated 1000 independent executions of Algorithm \ref{alg:online-detection} and collected the frequency at which the algorithm correctly identified the circuit structure. Results are displayed in subplots A and B of Figure \ref{fig:multi_plot}.

Subplot A shows the probability of failing to reject the null hypothesis given the null hypothesis is true, which should be exactly $\alpha$. The numerics aligned with the theoretical value with the exception of cases with low shot counts. This can be understood as our estimator is built upon asymptotic statements on the sampling distribution. Subplot B demonstrates the rate of true negatives. We can see that, given sufficient samples, we always identified non-golden cuts correctly. For lower significance values, we are more prone to rejecting the null hypothesis, explaining the faster rate of convergence towards 1 for smaller $\alpha$-values.

We also examined the quality of the reconstruction by calculating the distance between the empirical and theoretical bitstring distributions. The theoretical distribution is obtained by taking large number of shots without circuit cutting. We employed the $\ell^2$-distance to quantify how far apart two distributions are, i.e., for discrete distributions $\bm p$ and $\bm q$,
\begin{align} 
    d(\bm p, \bm q) = \sqrt{\sum_{i} (p_i - q_i)^2}.
\end{align}

Again, we executed Algorithm \ref{alg:online-detection} independently for 1000 trials and collected the $\ell^2$ distance between the empirical, reconstructed bitstring distribution and the respective theoretical distribution at varying numbers of circuit execution shots and alpha levels. Results are found in subplots C and D of Figure \ref{fig:multi_plot}.

In general, the reconstruction error decreases monotonically with the number of shots, and there was not a significant difference among choices of $\alpha$. In the case of low shot count and no existing golden cut, there seems to be more statistical fluctuation when reconstructing. Considering subplot B above, we know that this region is prone to false negatives, and thus neglecting bases that should not be neglected. 

\subsection{Runtime Analysis}
\label{sec:numerics-time}

\renewcommand{\arraystretch}{1.5}
\begin{table}
\centering
\begin{tabular}{|c|c|c|c|}
    \hline
    \multicolumn{2}{|c|}{} & \multicolumn{2}{|c|}{run time (sec)} \\
    \hline
    \multicolumn{2}{|c|}{} & w/ optimization & w/o optimization \\
    \hline
    \multirow{3}{*}{$\alpha$} & $10^{-1}$ & 0.0771$\pm$0.0006 & 0.0959$\pm$0.0004 \\
    \cline{2-4}
    & $10^{-2}$ & 0.0749$\pm$0.0004 & 0.0961$\pm$0.0004 \\
    \cline{2-4}
    & $10^{-3}$ & 0.0747$\pm$0.0003 & 0.0962$\pm$0.0004 \\
    \hline
\end{tabular}
\caption{Runtime comparison between circuit cutting procedures with and without optimization from Algorithm \ref{alg:online-detection}. Independent trials were repeated 1000 times. We can see that neglecting basis elements consistently run faster despite spending computing overhead on hypothesis testing.}
\label{tab:runtime-tab} 
\end{table}

To obtain timed runtime values, we generated a circuit with a golden cut and executed Algorithm \ref{alg:online-detection}. Recall that, depending on the results of the hypothesis test on the upstream circuit, the downstream circuit might not be executed for certain bases. Then, we ran the same cut circuit and performed the usual reconstruction routine without hypothesis testing or golden cutting optimization. Both of these processes were timed individually over 1000 trials and at varying alpha levels. Results for this can be found in Table \ref{tab:runtime-tab}. In general, we see roughly a $20\%$ decrease in runtime upon performing the optimization. As $\alpha$ decreases, we tend to reject the null hypothesis more often, thereby executing the downstream circuit more often. 

\section{Conclusion}
\label{sec:conclusion}

In this manuscript, we proposed an online detection algorithm for finding golden cutting points---circuit structures that induce neglectable basis elements during circuit reconstruction. The detection algorithm was built on performing a hypothesis test for each basis element, and executing the respective downstream circuit only if the null hypothesis is rejected. The detection does not require additional circuit executions. We showed numerically that, under sufficient number of shots, golden cuts will always be detected and there is no drastic difference among choices of significance levels.

Empirically testing the algorithm on quantum hardware and at large scale---both of which contribute additional noise that can affect the quality of estimators---are left to future work. Another immediate open question is the prevalence of circuit structure amenable to golden cuts in applicable circuits. For instance, we found that the SupermarQ~\cite{supermarq} and QASMBench~\cite{qasmbench} benchmark suites both feature a handful of benchmarks that exhibit this circuit structure. Variational circuits whose ansatz can be flexible might also be a candidate to apply golden cutting point restrictions for better scalability.




\section*{Acknowledgements}
This material is based upon work supported by the U.S. Department of Energy, Office of Science, National Quantum Information Science Research Centers. The colors used in figures were based on Nord Theme~\cite{nord}. 

\bibliographystyle{IEEEtran}
\bibliography{IEEEabrv,ref}

\end{document}